\documentclass{article}

\usepackage{amsmath,amssymb,xcolor,graphicx,wrapfig,paralist}
\usepackage{algorithmicx,algorithm}
\usepackage[noend]{algpseudocode}
\usepackage{multirow}
\usepackage{booktabs}
\usepackage{proof}
\usepackage{tikz}
\usepackage{stmaryrd}

\usetikzlibrary{shapes,positioning,arrows,calc,automata,matrix}
\tikzstyle{state}+=[minimum size = 8mm, inner sep=0,outer sep=1]
\tikzset{->,>=stealth'}

\definecolor{wwhite}{gray}{1}
\usepackage[T1]{fontenc}

\usepackage{tabularx}
\newcolumntype{L}{>{\raggedright\arraybackslash}p{1.6cm}}
\newcolumntype{C}{>{\centering\arraybackslash}p{1.6cm}}
\newcolumntype{R}[1]{>{\raggedleft\arraybackslash}p{#1}}
\allowdisplaybreaks

\newcommand{\thmhelperpre}[2]{\newcommand{\theoremlike}[1]{\par\medskip\penalty-250\refstepcounter{theorem}{\bfseries\noindent##1 \ref{#1}.}\itshape}\theoremlike{#2}}
\newcommand{\thmhelperpost}{\par\medskip%
 \renewcommand{\theoremlike}[1]{\par\medskip\penalty-250\refstepcounter{theorem}{\bfseries\noindent##1 \thesection .\thetheorem.}\itshape}%
}
\newenvironment{reflemma}[1]{\thmhelperpre{#1}{Lemma}}{\thmhelperpost}

\usepackage{microtype}
\newcommand{\myspace}{{}}
\newcommand{\mybigspace}{{}}

\newcommand{\prob}{\mathrm{Pr}} 

\newcommand{\pr}{\mathbb P}

\newcommand{\pmin}{p_{\mathsf{min}}}

\newcommand{\Nset}{\mathbb N}

\newcommand{\dist}{\mathsf{D}}
\newcommand{\distance}[1]{\dist_{#1}}
\newcommand{\diste}{\mathsf{E}_\mathsf{FT}} 
\newcommand{\distv}{\dist_\mathsf{TV}} 
\newcommand{\dists}{\dist_\mathsf{FT}} 
\newcommand{\disti}{\dist_\mathsf{IT}} 

\newcommand{\lang}{\mathsf{L}} 
\newcommand{\pref}[1]{\mathsf{pref}({#1})}

\newcommand{\Ap}{Ap}
\newcommand{\Mc}{\mathcal{M}}
\newcommand{\Pm}{\mathbf{P}}
\newcommand{\St}{S}
\newcommand{\init}{\mu}
\newcommand{\Lab}{L}

\newcommand{\run}{\rho}
\newcommand{\ppath}{\pi}
\newcommand{\word}{w}

\newcommand{\runs}{\mathsf{Runs}}

\newcommand{\LTL}{\ensuremath{\mathsf{LTL}}}

\newcommand{\X}{{\ensuremath{\mathbf{X}}}}
\newcommand{\F}{{\ensuremath{\mathbf{F}}}}
\newcommand{\G}{{\ensuremath{\mathbf{G}}}}

\newcommand{\cone}{\mathsf{Cone}}
\newcommand{\field}{\mathcal F}
\newcommand{\lclass}{\mathcal L}

\newcommand{\trerr}[1]{\xi}
\newcommand{\mperr}[1]{\zeta}

\usepackage{amsthm}
\newtheorem{theorem}{Theorem}
\newtheorem{lemma}{Lemma}
\theoremstyle{definition}
\newtheorem{definition}{Definition}
\newtheorem{example}{Example}
\newtheorem{remark}{Remark}

\theoremstyle{remark}
\newtheorem*{claim}{Claim}

\title
{Linear Distances between Markov Chains}

\usepackage{authblk}
\author{Przemys\l aw Daca$^1$~~~ 
Thomas A.\ Henzinger$^1$ \\
Jan K{{\v r}}et\'insk\'y$^2$~~~ 
Tatjana Petrov$^1$
}
\affil{$^1$IST Austria, Klosterneuburg, Austria\\
$^2$Institut f\"ur Informatik, Technische Universit\"at M\"unchen, Germany
}
\date{}

\begin{document}

\maketitle

\begin{abstract}
	We introduce a general class of distances (metrics) between Markov chains, which are based on linear behaviour. 
	This class encompasses distances given topologically (such as the total variation distance or trace distance) as well as by temporal logics or automata. 
	We investigate which of the distances can be approximated by observing the systems, i.e. by black-box testing or simulation, and we provide both negative and positive results.
\end{abstract}

\section{Introduction}

Behaviour of processes is traditionally compared using various notions
of equivalence, such as trace equivalence, bisimulation, etc.
However, the concept of equivalence is often too coarse for quantitative
systems, such as Markov chains.  For instance, probabilities of
failures of particular hardware components
are typically only empirically estimated and the slightest imprecision
in the estimate may result in breaking the equivalence between
processes.  Moreover, if the (possibly black-box) processes are indeed different we would
like to measure \emph{how much they differ}.
This has led to lifting the Boolean idea of behavioural equivalence to
a finer, quantitative notion of behavioural \emph{distance} between processes.  The
distance between processes $s$ and $t$ is typically formalized as
$\sup_{p\in \mathcal C} |p(s)-p(t)|$ where $\mathcal C$ is a class of
properties of interest and $p(s)$ is a quantitative value of the
property $p$ in process $s$ \cite{DBLP:conf/concur/DesharnaisGJP99}.  This notion has been introduced
in~\cite{DBLP:conf/concur/DesharnaisGJP99} for Markov chains and
further developed in various settings, such as Markov decision
processes \cite{DBLP:conf/aaai/FernsPP04}, quantitative transition
systems~\cite{DBLP:conf/lics/AlfaroMRS07}, or concurrent
games~\cite{DBLP:conf/icalp/AlfaroFS04}.

Several kinds of distances have been investigated for Markov chains.
On the one hand, branching distances,
e.g.~\cite{DBLP:journals/entcs/Abate13,DBLP:conf/concur/DesharnaisGJP99,DBLP:journals/tcs/BreugelW06,DBLP:conf/fossacs/BreugelSW07,DBLP:conf/tacas/BacciBLM13,DBLP:conf/mfcs/BacciBLM13,DBLP:conf/qest/BacciBLM13,DBLP:journals/ejcon/GirardP11},
lift the equivalence given by the probabilistic bisimulation of Larsen
and Skou \cite{DBLP:conf/popl/LarsenS89}.  On the other hand, there
are \emph{linear} distances, in particular the total variation distance
\cite{DBLP:conf/csl/ChenK14,DBLP:conf/fossacs/BacciBLM15} and trace
distances \cite{DBLP:conf/qest/JaegerMLM14,DBLP:conf/ictac/BacciBLM15}.  Linear distances are
particularly appropriate when (i) we are interested in linear-time
properties, and (ii) we want to estimate the distance based only on
simulation runs from the initial distribution of the system, i.e.\ in a black-box  setting.
(Recall that for branching distances, the underlying probabilistic
bisimulation corresponds to testing equivalence where not only runs
from the initial distribution can be observed, but it is also possible to dump the
current state of the system, and later restart the simulation from this state~\cite{DBLP:conf/popl/LarsenS89}.)

In this paper, we introduce a simple framework for linear distances
between Markov chains, using the formula above, where $p(s)$ is the
probability of satisfying $p$ when starting a simulation run in state $s$ (when $p$ is seen
as a language of $\omega$-words it is the probability to generate a trace belonging to
$p$).
We consider several classes $\mathcal C$ of languages of interest,
characterized from several points of view, e.g.\ topologically, by linear-time logics,
or by automata, thus rendering our framework versatile.

We investigate when a given distance can be estimated in a black-box setting, i.e.\ only from simulations.
One of the main difficulties is that the class $\mathcal C$ typically
includes properties with arbitrarily long horizon or even
infinite-horizon properties, whereas every simulation run is
necessarily finite.  Note that we do not employ any simplifications
such as imposed fixed horizon or discounting, typically used for
obtaining efficient algorithms, e.g.,
\cite{DBLP:conf/concur/DesharnaisGJP99,DBLP:journals/tcs/BreugelW06,DBLP:conf/mfcs/BacciBLM13},
and the undiscounted setting is fundamentally more
complex~\cite{DBLP:conf/fossacs/BreugelSW07}.
Since even simpler tasks are impossible for unbounded horizon in the black-box setting without any further knowledge, we assume we only know a lower bound on the minimum transition probability $\pmin$.
Note that knowledge of $\pmin$ has been justified in \cite{DBLP:journals/corr/DacaHKP15}.

Our contribution is the following:
\myspace
\begin{itemize}
\item We introduce a systematic linear-distance framework and illustrate it
  with several examples, including distances previously investigated in the literature.
\item The main technical contributions are (i) a negative result stating
  that the total variation distance cannot be estimated by simulating the
  systems, and (ii) a positive result that the trace distance can be
  estimated.
\item These results are further exploited to provide both negative and positive results for each of the settings where the language class is given topologically, by LTL (linear temporal logic) fragments, and by automata.
We also show that the negative result on the total variation distance can be turned into a positive result if the transition probabilities have finite precision.
\end{itemize}
\myspace

\subsection{Related work}
There are two main linear distances considered for Markov chains:
the total variation distance and trace distance. Several algorithms have been
proposed for both of them in the case when the Markov chains are
known (white-box setting).  We are not aware of any work where the distances are estimated
only from simulating the systems (black-box setting).

Firstly, for the \emph{total variation distance} in the white-box setting,
\cite{DBLP:conf/csl/ChenK14} shows that deciding whether it equals one can
be done in polynomial time, but computing it is NP-hard and not known
to be decidable, however, it can be approximated;
\cite{DBLP:conf/fossacs/BacciBLM15} considers this distance more
generally for semi-Markov processes, provides a different
approximation algorithm, and shows it coincides with distances based
on (i) metric temporal logic, and (ii) timed automata languages.

Secondly, the \emph{trace distance} is based on the notion of trace
equivalence, which can be decided in polynomial time
\cite{DBLP:journals/ijfcs/DoyenHR08} (however, trace refinement of
Markov decision processes is already undecidable
\cite{2015arXiv151009102F}). Several variants of trace distance are
considered in \cite{DBLP:conf/qest/JaegerMLM14} where it is taken as a
limit of finite-trace distances, possibly using discounting or
averaging.  In \cite{DBLP:conf/ictac/BacciBLM15} the finite-trace
distance is shown to coincide with distances based on (i) LTL, and (ii)
LTL without the U operator, i.e., only using the X operator and Boolean
connectives. This distances is also shown to be NP-hard and not known to be decidable,
similarly to the total variation distance. Finally, an approximation
algorithm is shown (again in the white-box setting), where the over-approximates are branching-time
distances, showing an interesting connection between the branching and linear
distances.

In \cite{KieferS16} the distinguishability problem is considered, i.e.\
given two Markov chains whether there is a monitor that reads a single
sample and with high probability decides which chain produced the
sequence. This is indeed possible when the total
variation distance between the chains equals one, and~\cite{KieferS16} shows how to
construct such monitors. In contrast, our negative results shows that
it is not possible to decide with high probability whether the total
variation distance equals one when the two Markov are black-box.

Linear distances have been proposed also for quantitative transition
systems, e.g.~\cite{DBLP:conf/icalp/AlfaroFS04}.   
Moreover, there are
other useful distances based on different fundaments; for instance,
the Skorokhod
distance~\cite{DBLP:conf/emsoft/CaspiB02,DBLP:conf/hybrid/MajumdarP15,DBLP:conf/cav/DeshmukhMP15}
measures the discrete differences between systems while allowing for
timing distortion;
Kullback-Leibler divergence \cite{DBLP:conf/qest/JaegerMLM14} is useful from the information-theoretic point of view.
Finally, distances have been also studied with respect to applications in linear-time model checking \cite{DBLP:conf/hybrid/TkachevA14,DBLP:conf/ictac/BacciBLM15}.
 
\myspace 
 
\subsection{Outline}
\myspace

After recalling the basic notions in Section~\ref{sec:prelim}, we
introduce our framework and illustrate it with examples in
Section~\ref{sec:framework}. We define our problem formally in Section~\ref{sec:problem}. In Sections~\ref{sec:neg} and~\ref{sec:pos}
we provide the proofs of our technically principal negative and
positive result, respectively.
Section~\ref{sec:disc} extends the results in
the settings of topology, logics and automata, and discusses general conditions for estimability.  
We conclude in Section~\ref{sec:concl}.


Technical proofs omitted in the text can be found in Appendix.

\myspace

\section{Preliminaries}
\label{sec:prelim}
\myspace

We consider a finite set $\Ap$ of atomic propositions and denote $\Sigma=2^{\Ap}$.
\begin{definition} [Markov chain]
A \emph{(labelled) Markov chain (MC)} is a tuple $\Mc = (\St, \Pm, \init, \Lab)$, where
\begin{itemize}
\item $\St$ is a finite set of \emph{states},
\item $\Pm \;:\; \St \times \St \to [0,1]$ is a \emph{transition} probability matrix, such that for every $s\in \St$ it holds $\sum_{s'\in \St} \Pm(s,s') = 1$,
\item $\init$ is an \emph{initial} probability distribution over $\St$, 
\item $\Lab : \St \to\Sigma$ is a \emph{labelling} function.
\end{itemize}
\end{definition}
A \emph{run} of $\Mc$ is an infinite sequence $\run = s_1 s_2 \cdots$
of states, such that $\init(s_1)>0$ and
$\Pm(s_i, s_{i+1}) > 0$ for all $i\geq 1$; we let $\run[i]$ denote the
state $s_i$.  A \emph{path} in $\Mc$ is a finite prefix of a run of
$\Mc$.  An \emph{$\omega$-word} is an infinite sequence $a_1 a_2 \cdots\in\Sigma^\omega$ of
symbols from $\Sigma$; a \emph{word} is a finite prefix $w\in\Sigma^*$ of an
$\omega$-word.  
We extend the labelling notation so that for a path $\pi\in S^k$, the projected sequence $L(\pi)$ is the word $w\in \Sigma^k$, where $w[i]=L(\pi[i])$, and the inverse map is $L^{-1}(w)= \{\pi\in S^k\mid L(\pi)=w\}$.
Given a path $\ppath = s_1 \cdots s_n$, we denote
the \emph{$k$-prefix} of $\ppath$ by $\ppath\downarrow k = s_1\cdots s_k$,
and 
similarly for prefixes of words.

Each path $\ppath$ in $\Mc$ determines the set of runs $\cone(\ppath)$
consisting of all runs that start with $\ppath$.  To $\Mc$ we assign
the probability space $
(\runs,\field,\pr_\Mc)$,
where $\runs$ is the set of all runs in $\Mc$, $\field$ is the
$\sigma$-algebra generated by all $\cone(\ppath)$, 
and $\pr_{\Mc}$ is the
unique probability measure such that
$\pr_{\Mc}(\cone(s_1\cdots s_n)) = \mu(s_1)\cdot\prod_{i=1}^{n-1}
\Pm(s_{i},s_{i+1})$,
where the empty product equals $1$. 
We will omit the subscript in $\pr_\Mc$ if the Markov chain is clear
from the context.
Further, we write $\pr^s_\Mc$ for the
probability measure, where $\init(s)=1$ and $\init(s')=0$ for
$s' \neq s$. 
Finally, we overload the notation and for a path $\pi$ write $\pr(\pi)$ meaning $\pr(\cone(\pi))$, and for a ($\omega$)-word $w$, we write $\pr(w)$ meaning $\pr(L^{-1}(w))$.





\myspace

\section{Framework for Linear Distances}
\label{sec:framework}
\myspace

In this section we introduce our framework for linear distances.
For $i\in\{1,2\}$, let $\Mc_i = (\St, \Pm_i, \init_i, \Lab)$ denote a Markov chain\footnote{To avoid clutter, the chains are defined over the same state space with the same labelling, which can be w.l.o.g.\ achieved by their disjoint union.} and $(\runs,\field,\pr_i)$ the induced probability space.
Since single runs of Markov chains typically have measure $0$, we introduce linear distances using measurable sets of runs:
\begin{definition}[$\lclass$-distance]
	For a class $\lclass\subseteq\field$ of measurable $\omega$-languages\footnote{Formally, the measurable space of $\omega$-languages is given by the set $\Sigma^\omega$ equipped with a $\sigma$-algebra $\field(\Sigma)$ generated by the set of cones $\{w\Sigma^\omega\mid w\in\Sigma^*\}$. This ensures, for every measurable $\omega$-language $X$, that $L^{-1}(X)$ is measurable in every MC.}, the $\lclass$-distance $\distance{\lclass}$ is defined by
	\[ \distance{\lclass}(\Mc_1, \Mc_2) = \sup_{X\in\lclass} |\pr_1(X) - \pr_2(X)|\ .\]
\label{dfn:2}\mybigspace
\end{definition}
Note that every $\distance{\lclass}$ is a pseudo-metric\footnote{It is symmetric, it satisfies the triangle inequality, and the distance between identical MCs is $0$.}. 
However, two different MCs can have distance $0$, for instance, when they induce the same probability space.

The definition of $\lclass$-distances  can be instantiated either
(i) by a direct topological description of $\lclass$, or indirectly 
(ii) by a class $\mathcal A$ of automata inducing the class of recognized languages $\lclass=\{\lang(A)\mid A\in\mathcal A\}$, 
or 
(iii) by a set of formulae $\mathfrak L$ of a linear-time logic inducing the languages of models $\lclass=\{\lang(\varphi)\mid \varphi\in\mathfrak L\}$
where $\mathsf L(\varphi)$ denotes the language of $\omega$-words satisfying the formula $\varphi$.	
	
We now discuss several particularly interesting instantiations:
\begin{example}[Total variation]
	One extreme choice is to consider all measurable languages, resulting in the \emph{total variation distance} 
	$\distv(\Mc_1, \Mc_2) = \sup_{X\in \field(\Sigma)} |\pr_1(X) - \pr_2(X)|$.
\end{example}

\begin{example}[Trace distances]\label{ex:trac-dist}
	The other extreme choices are to consider (i) only the generators of $\field(\Sigma)$, i.e.\ the cones $\{w\Sigma^\omega\mid w\in\Sigma^*\}$, resulting in the \emph{finite-trace distance} $\dists(\Mc_1, \Mc_2) = \sup_{w \in \Sigma^+} |\pr_1(w) - \pr_2(w)|$; 
	or (ii) only the elementary events, i.e. $\Sigma^\omega$, resulting in the \emph{infinite-trace distance} $\disti(\Mc_1, \Mc_2) = \sup_{w \in \Sigma^\omega} |\pr_1(w) - \pr_2(w)|$.
\end{example}

\begin{example}[Topological distances]
	There are many possible choices for $\lclass$ between the two extremes above, such as \emph{clopen sets} $\Delta_1$, which 
	are finite unions of cones (being both closed and open), \emph{open sets} $\Sigma_1$, which are infinite unions of cones, \emph{closed sets} $\Pi_1$, or classes higher in the \emph{Borel hierarchy} such as the class of \emph{$\omega$-regular} languages (within $\Delta_3$), or languages given by thresholds for a \emph{long-run average reward} (within $\Sigma_3$).
\end{example}

\begin{example}[Automata distances]\label{ex:auto-dist}
	The class of $\omega$-regular languages can also be given in terms of automata, for instance by the class of all deterministic \emph{Rabin automata} (DRA). Similarly, the closed sets $\Pi_1$ correspond to the class of deterministic B\"uchi automata with all states final.
	Further, we can restrict the class of all DRA to those of \emph{size at most $k$} for a fixed $k\in\Nset$, denoting the resulting distance by $\dist_{DRA\leq k}$.  
\end{example}

\begin{example}[Logical distances]\label{ex:log-dist}
	The class of $\omega$-regular languages can also be given in terms of logic, by the monadic second-order logic (with order).
	Further useful choices include the \emph{first-order logic with order}, corresponding to the star-free languages and to the \emph{linear temporal logic} (LTL), or its fragments such as LTL with only $\X$ or only $\F$ and $\G$ operators etc.
\end{example}

%
%
%


\begin{remark}
	The introduced distances can also be considered in the discrete setting, resulting in various notions of equivalence. 
	For instance, the \emph{finite-trace equivalence} $\diste$ can be derived from the finite-trace distance by the following discretization:\myspace
	\[ \diste(\Mc_1, \Mc_2) =
	\begin{cases}
	0 & \text{ if } \dists(\Mc_1, \Mc_2)=0 \\
	1 & \text{ otherwise, i.e., } \dists(\Mc_1, \Mc_2)>0. \\
	\end{cases}
	\]
\end{remark}
\mybigspace

\section{Problem Statement}\label{sec:problem}
\myspace

Linear distances can be very useful when we want to compare a black-box system with another system, e.g. a white-box specification or a black-box previous version of the system.
Indeed, in such a setting we can typically obtain simulation runs of the system and we must establish a relation between the systems based on these runs only.
This is in contrast  with branching distances where either both systems are assumed white-box or there are strong requirements on the testing abilities, such as dumping the current state of the system, arbitrary many restarts from there, and nesting this branching arbitrarily.
Therefore, we focus on the setting where we can obtain only finite prefixes of runs and we use statistics to (i) deduce information on the whole infinite runs, and (ii) estimate the distance of the systems.

\medskip
\noindent\framebox[\linewidth]{\parbox{0.95\linewidth}{
    For a given distance function $\distance{\lclass}$, the goal is to construct an almost-surely terminating algorithm that given
\myspace

\begin{itemize}
\item any desired finite number of sampled simulation run from Markov chains $\Mc_1$ and $\Mc_2$ of any desired finite length,
\item lower bound $\pmin>0$ on the minimum (non-zero) transition probability,
\item confidence $\alpha\in (0,1)$,
\item interval width $\delta\in (0,1)$, 
\end{itemize}
computes an interval $I$ such that $|I|\leq\delta$ and
$\prob[ \distance{\lclass}(\Mc_1, \Mc_2) \in I] \geq 1 - \alpha$.
}}

\medskip
A distance function is called \emph{estimable}, if there exists an algorithm in the above sense, and \emph{inestimable} otherwise.

\myspace
\section{Inestimability: Total variation distance}
\label{sec:neg}
\myspace We show that for the total variation distance $\distv$ there
exists no ``statistical'' algorithm (in the above sense) which is
correct for all inputs $(\Mc_1,\Mc_2, \alpha, \delta)$.
Our argument consists of two steps:
\begin{enumerate}
\item We construct two chains such that $\distv(\Mc_1,\Mc_2)=1$, namely the two MCs shown in Figure~\ref{fig:mc-cex} (similar to \cite{DBLP:conf/qest/JaegerMLM14}): one with $\tau=0$ and the other with small $\tau>0$.
\item We show that any potentially correct algorithm will give with high probability an incorrect output for some choice of $\tau,\alpha,\delta$.
\end{enumerate}

\begin{figure}[h]
\centering\mybigspace
\scalebox{1}{
\begin{tikzpicture}[node distance=2cm]
\node[state,initial text=] (s) at (0,0){$a$};
\node[state] (t) [right of=s] {$b$};
\path[->]
(s) edge [bend right] node[below]{$0.5+\tau$} (t)
(t) edge [bend right] node[above]{$0.5-\tau$} (s)
(s) edge [loop left] node[left]{$0.5-\tau$} (s)
(t) edge [loop right] node[right]{$0.5+\tau$} (t);
\end{tikzpicture}
}\myspace\myspace
\caption{A Markov chain with labelling displayed in states.}
\label{fig:mc-cex}
\vspace*{-0.5em}
\end{figure}
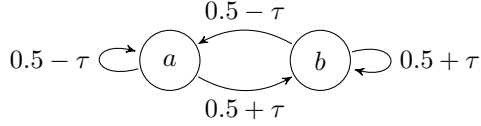

\noindent \textbf{Maximizing event} We start by showing that even
an arbitrarily small difference in transition probabilities between two
Markov chains may result in total variation distance of $1$. Consider
two Markov chains as in Figure \ref{fig:mc-cex}, where $\Mc_1$ has
$\tau=0$, and $\Mc_2$ has $\tau>0$. We assume that the initial
distribution for each chain is its stationary distribution. In this setting,
every simulation step is like an independent trial with probability
$0.5-\tau$ (resp.\ $0.5+\tau$) of seeing $a$ (resp.\ $b$).

Let $X_n$ (resp. $Y_n$) denote the number of $b$ symbols in a random
path of length $n$ sampled from $\Mc_1$ (resp. $\Mc_2$). 
By the central limit theorem the distributions of $X_n$ and $Y_n$
are converging to the normal distribution when $n\rightarrow \infty$:
\[ X_n \approx \mathcal{N}(0.5n, 0.5^2n)\qquad\qquad Y_n \approx \mathcal{N}((0.5+\tau)n,n(0.25-\tau^2)).\]

For $n\in \Nset$ let the event $E_n$ mean
``there is at most $c_n = 
(0.5+\tau/2)n$
 symbols $b$ in the path
prefix of length $n$.'' 
%
The probabilities of event $E_n$ in the two Markov chains are:
\[ \pr_{\Mc_1}(E_n) = \pr_{\Mc_1}(X_n \leq c_n) = \Phi(\tau\sqrt{n})\qquad\pr_{\Mc_1}(E_n) = \pr_{\Mc_1}(Y_n \leq c_n) = \Phi(\frac{-0.5\tau\sqrt{n}}{\sqrt{0.25-\tau^2}}),\]
where $\Phi$ is the CDF of the standard normal distribution. For
$n\rightarrow \infty$ the probability of $E_n$ in $\Mc_1$ and $\Mc_2$ converges
to $1$ and $0$, respectively, so the total variation
distance converges to $1$.

\medskip
\noindent
\textbf{Negative result for total variation distance}
Now we show that there is no statistical procedure for estimating
total variation distance that would almost-surely terminate.

\begin{theorem}\label{thm:tv}
  For any $\delta<1$ and $\alpha< \frac{1}{2}$, there is no algorithm for
  computing a $1-\alpha$ confidence interval of size $\delta$ for the total
  variation distance that almost-surely terminates.
\end{theorem}


\begin{proof}
  Let us write $\Mc(\tau)$ for a Markov chain in Figure
  \ref{fig:mc-cex} with the parameter $\tau$ and 
  the initial distribution being stationary.

  For $\alpha < \frac{1}{2}$ we define the following decision problem
  $\mathsf{B}_\alpha$:
  \begin{itemize}
  \item The input to $\mathsf{B}_\alpha$ is a single path from $\Mc(\tau)$ of
  arbitrary length, where $\tau$ is unknown,
\item The task of $\mathsf{B}_\alpha$ is to output answer \textbf{Yes}
  with probability $\geq 1-\alpha$ if
  $\distv(\Mc(0), \Mc(\tau)))=1$, output answer \textbf{No}
  with probability $\geq 1-\alpha$ if
  $\distv(\Mc(0), \Mc(\tau))=0$. Note that $\distv(\Mc(0), \Mc(\tau))$ can equal only $0$ or $1$.
  \end{itemize}

  The remaining part of proof is done in two parts. In the first part, we show that there is no
  algorithm that solves $\mathsf{B}_\alpha$ and almost-surely terminates. In the second part we
  reduce the problem $\mathsf{B}_\alpha$ to computing a confidence
  interval for the total variation distance.

\medskip

  \noindent \textbf{Part I.}
  Suppose the opposite of the claim: that for some $\alpha<\frac{1}{2}$ there
  is an algorithm which solves $\mathsf{B}_\alpha$ and almost-surely
  terminates.  We represent the algorithm for solving
  $\mathsf{B}_\alpha$ as a deterministic Turing machine $\mathsf{TM}$,
  which works as follows:
  \begin{enumerate}
  \item The input tape of $\mathsf{TM}$ contains a (single) randomly sampled run of $\Mc(\tau)$,
  \item $\mathsf{TM}$ reads a part of the run from the tape and eventually returns \textbf{Yes}/\textbf{No} answer.
  \end{enumerate}

  The input to the $\mathsf{TM}$ is random, therefore we can assign a
  probability distribution to the computations of $\mathsf{TM}$.  To this
  end, we represent the answer of $\mathsf{TM}$ by the random variable
  $X : \runs \mapsto \{\text{\textbf{Yes}, \textbf{No}}\}$, and we use
  the random variable $Y : \runs \mapsto \Nset\cup \{\infty\}$ to
  represent the number of path symbols $\mathsf{TM}$ reads before
  terminating, where $\infty$ means that $\mathsf{TM}$ does not
  terminate.

  Suppose we run $\mathsf{TM}$ on the Markov chain $\Mc(0)$. We
  write $\pr_1$ for the probability measure of $\mathsf{TM}$ on this
  input.  The total variation distance between the two Markov chains
  $\Mc(0)$ is $0$, so with probability $\geq 1-\alpha$ $\mathsf{TM}$
  returns answer \textbf{No}, i.e.\ $\pr_\mathsf1(X=\text{\textbf{No}}) \geq 1-\alpha.$

  By assumption $\mathsf{TM}$ almost-surely terminates on every input,
  so $\pr_1(Y \in \Nset) = 1$.  Let $q$ be the following quantile:
  \[ q = \min\{ c\in \Nset~:~\pr_1(Y \leq c)\geq 0.5+\alpha \}. \]
  \begin{claim}
  	$q\in \Nset$.
  \end{claim}
 It follows that:
  \begin{equation}\label{eq:n-1}
     \pr_1(X=\text{\textbf{No}} \land Y \leq q)= 1 - \pr_1(X=\text{\textbf{Yes}} \vee Y > q)  
     \geq 1 - \pr_1(X=\text{\textbf{Yes}}) - \pr_1(Y > q) \geq 0.5.
  \end{equation}


  Turing machine $\mathsf{TM}$ is deterministic, so if it terminates after reading 
  prefix $\ppath$ of some run $\run$, then it terminates after reading prefix $\ppath$ of any run.
  As a consequence, the event $Y\leq q$ can be represented as a union of $\ell$ cones where $\ell \leq |\Sigma|^q = 2^q$ since $\Sigma=\{a,b\}$ in $\Mc$:
  \[\{\run ~:~ Y(\run) \leq q\} = \bigcup_{i=1}^\ell \cone(\ppath_i),\]
  where all $\ppath_i\in \Sigma^q$ are distinct.  The event
  $X=\text{\textbf{No}} \land Y \leq q$ is a refinement of the event
  $Y\leq q$, so it may also be represented as
  \begin{equation}
\label{eq:n0}
\{\run ~:~ X=\text{\textbf{No}} \land Y(\run) \leq q\} = \bigcup_{i=1}^m \cone(\ppath_i),
  \end{equation}
  where $m\leq \ell\leq 2^q$.
  Since every path in $\Mc(0)$ of length $q$ has probability $0.5^q$, we get by \eqref{eq:n0}
  \[ \pr_1(X=\text{\textbf{No}} \land Y(\run) \leq q) =
  \pr_1(\bigcup_{i=1}^m \cone(\ppath_i)) = \sum_{i=1}^m \pr_1(\ppath_i)
  = m0.5^q.\] Then by \eqref{eq:n-1} it follows that $m \geq 2^{q-1}$.

  Now, we run $\mathsf{TM}$ on the Markov chain
  $\Mc(\epsilon)$ where 
  \( \epsilon =
0.5 - \alpha^\frac{1}{q}2^\frac{1-q}{q} 
   \)
   if $q>0$ and $\epsilon=0.25$ in the degenerated case of $q=0$.
   \begin{claim}
   $\epsilon>0$.
   \end{claim}

  Let us write $\pr_2$ for the probability measure of $\mathsf{TM}$ on
  the input $\Mc(\epsilon)$.
The distance between $\Mc(0)$ and $\Mc(\epsilon)$ is
  $1$, since $\epsilon>0$. As a consequence, $\mathsf{TM}$ should return answer \textbf{Yes} on this
  input with probability $\geq 1-\alpha$, or equivalently answer
  \textbf{No} with probability $< \alpha$. We show, however, that the probability of \textbf{No} is $\geq \alpha$:
  \begin{align*}
    &\pr_2(X=\text{\textbf{No}} \land Y \leq q) = \sum_{i=1}^m \pr_2(\ppath_i) &\text{by \eqref{eq:n0}} \\
    &  = \sum_{i=1}^m (0.5+\epsilon)^{u_i}(0.5-\epsilon)^{q-u_i}  & u_i \text{ is number of $b$'s in }\ppath_i\\
    &\geq \sum_{i=1}^m (0.5-\epsilon)^q = m(0.5-\epsilon)^q\\
    &\geq 2^{q-1}(0.5-\epsilon)^q = \alpha. &\text{by }m \geq 2^{q-1}..
  \end{align*}
  We obtain a contradiction, thus the assumed machine $\mathsf{TM}$ does not exist.

\medskip

\noindent  \textbf{Part II.}  Suppose for a contradiction that for some
  $\alpha< \frac{1}{2},\delta<1$ there exists an algorithm
  $\mathsf{Alg}_{\alpha,\delta}$ that solves the problem defined in the
  theorem and almost-surely terminates. Then then this algorithm can
  solve the problem $\mathsf{B}_\alpha$ in the following way:
  \begin{enumerate}
  \item Use $\mathsf{Alg}_{\alpha,\delta}$ to compute a confidence interval
    $I$ for the total variation distance between $\Mc(0)$ and
    $\Mc(\tau)$.  Algorithm $\mathsf{Alg}_{\alpha,\delta}$ can sample any
    number of paths from $\Mc(0)$.  Observe that in $\Mc(\tau)$
    probability of seeing states $a$ and $b$ remains constant over
    time.  Thus, sampling multiple paths from $\Mc(\tau)$ by
    $\mathsf{Alg}_{\alpha,\delta}$ can be replaced by sampling a single path
    from $\Mc(\tau)$.
  \item Output \textbf{Yes} if $1\in I$, \textbf{No} if $0\in I$.
  \end{enumerate}
  We have shown that for any $\alpha<\frac{1}{2}$ the problem
  $\mathsf{B}_\alpha$ cannot by solved by an algorithm that
  almost-surely terminates.  As a consequence, the algorithm
  $\mathsf{Alg}_{\alpha,\delta}$ cannot exist.
\end{proof}

From Part II, it follows that there is no statistical algorithm even for fixed $\alpha$ and $\delta$.

\myspace

\section{Estimability: Finite-trace distance}
\label{sec:pos}
\myspace
In Section~\ref{ssec:fixed} we show how to estimate the distance given by traces of a fixed length. In Section~\ref{ssec:unbounded} we show how to reduce the problem of computing the finite-trace distance $\dists$ (where traces of arbitrary lengths are considered) to computing a constant number of fixed-length distances. 
\myspace

\subsection{Estimates for fixed length}\label{ssec:fixed}
\myspace

Given two Markov chains $\Mc_1$ and $\Mc_2$ we wish to estimate the finite-trace distance for fixed length $k\in \Nset$
\[ \dists^k = \sup_{w \in {\Sigma}^k} |\pr_1(w) - \pr_2(w)|.\]
There is $m=|\Sigma|^k$ words in $\Sigma^k$ (we enumerate them as $w_1,\cdots,w_m$), so the traces of 
length $k$ follow a multinomial distribution, i.e.\ for $i=1,2$
$\sum_{j=1}^{m}, \pr_i(w_j) = 1$.

We present a statistical procedure that estimates $\dists^k$ with arbitrary precision.
For $j\leq |\Sigma|^k$ we call a \emph{contrast} $\Delta_j$ the difference in probabilities
of trace $w_j$ between $\Mc_1$ and $\Mc_2$:
$\Delta_j = |\pr_1(w_j) - \pr_2(w_j)|$.  The distance $\dists^k$ is
the maximum over all such contrasts
$\dists^k = \max_{j\leq m} \Delta_j$.
We use the statistical procedure of \cite{goodman1964} to
simultaneously estimate all contrasts.  We sample random paths from
both Markov chains, and let $n_i^j$ denote the number of observations
of trace $w_j$ in a Markov chain $\Mc_i$. We write
$n_i = \sum_{j\leq m} n_i^j$ for the sum of all observations in $\Mc_i$.  The
estimator of $\pr_i(w_j)$ is $\tilde{p}_i^j = \frac{n_i^j}{n_i}$, and
the estimator of $\Delta_j$ is
$\tilde{\Delta}_j = |\tilde{p}_1^j - \tilde{p}_2^j|$.

\begin{theorem}[\cite{goodman1964}]
\label{thm:multi}
As $n_1, n_2\rightarrow \infty$ the probability approaches $1-\alpha$ that simultaneously for all contrasts 
\mybigspace
\[ |\Delta_j - \tilde{\Delta}_j| \leq S_j M\qquad\text{where}\qquad S_j = \sqrt{\frac{\tilde{p}_1^j - (\tilde{p}_1^j)^2}{n_1} + \frac{\tilde{p}_2^j - (\tilde{p}_2^j)^2}{n_2}},\]
\noindent
and $M$ is the square root of the $\frac{1-\alpha}{100}$ percentile of the $\chi^2$ distribution with $|\Sigma|^k$ degrees of freedom. 
\end{theorem}
The procedure for estimating $\dists^k$ works as follows.  For
$\epsilon,\alpha >0$ we sample paths from $\Mc_1$ and $\Mc_2$ until,
by Theorem~\ref {thm:multi}, with probability $1-\alpha$ for all
contrasts $|\Delta_j - \tilde{\Delta}_j| \leq \epsilon$.
Then with probability $1-\alpha$ it holds that
\( |\dists^k - \max_{j\leq m} \tilde{\Delta}_j| \leq \epsilon. \)


\myspace

\subsection{Estimates for unbounded length}\label{ssec:unbounded}
\myspace

Intuitively, the longer the path, the less probable it is, and the less distance it can cause.
However, this is only true if along the path probabilistic choices are made repeatedly. 

\begin{definition}
In a Markov chain $\Mc$, a 
state $s\in S$ is $k$\emph{-deterministic},
if there exists a word $w$ of length $k$, such that $ \pr^s(w) = 1$.
Otherwise, $s$ is $k$\emph{-branching}.  
A state $s\in S$ is 
\emph{deterministic}, if it is $k$-deterministic for all $k\in{\mathbb N}$.
\end{definition}
\myspace

\begin{lemma}
	\label{lem:det-branch}
If $s\in S$ is $k$-branching, it is also $(k+1)$-branching. 
Dually, if it is $k$-deterministic, it is also $(k-1)$-deterministic.
\end{lemma}
\myspace

%

\noindent
\begin{minipage}{0.65\textwidth}
\begin{example}
Every state is trivially $1$-deterministic. 
In Figure~\ref{fig:ex_with_thm2_b}, the leftmost state is $3$-deterministic and $4$-branching.
The states of the MC on the right are deterministic.

\end{example}
\end{minipage}
~~~~
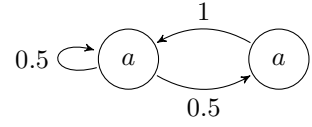
\begin{minipage}{0.3\textwidth}
 \begin{tikzpicture}[node distance=2cm]
 \node[state,initial text=] (s) at (0,0){$a$};
 \node[state] (t) [right of=s] {$a$};
 \path[->]
 (s) edge [loop left] node[left]{$0.5$} (s)
 (s) edge [bend right] node[below]{$0.5$} (t)
 (t) edge [bend right] node[above]{$1$} (s);
 \end{tikzpicture}
\end{minipage}


\begin{lemma} 
\label{lemma:det}
Consider a state $s$ in a Markov chain $\Mc$ with $n$ states.
If state $s$ is $n^2$-deterministic, then it is 
deterministic.
\label{lem:br}
\end{lemma}
Before proceeding to the proof, notice that even though 
it may seem that every branching state must be $n+1$ branching, this is not the case in general. Observe the counterexample in Fig.~\ref{fig:ex_with_thm2_a}. 
The leftmost state is $6$-deterministic (only the word $aaabaa$ can be generated), while $n=4$.

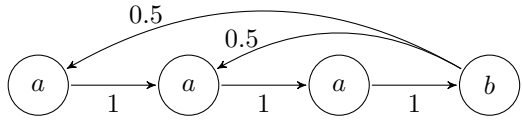
\begin{figure}[h]
	\centering\mybigspace
	\scalebox{1}{
		\begin{tikzpicture}[node distance=2cm]
		\node[state,initial text=] (s) at (0,0){$a$};
		\node[state] (t) [right of=s] {$a$};
		\node[state] (u) [right of=t] {$a$};
		\node[state] (w) [right of=u] {$b$};
		\path[->]
		(s) edge [right] node[below]{$1$} (t)
		(t) edge [right] node[below]{$1$} (u)
		(u) edge [right] node[below]{$1$} (w)
		(w) edge [bend right] node[above,pos=0.8]{$0.5$} (s)	
		(w) edge [bend right] node[above,pos=0.9]{$0.5$} (t);
		\end{tikzpicture}
	}\myspace
	\caption{Markov chain with $4$ states. The leftmost state is $6$-deterministic, but not deterministic.}
	\label{fig:ex_with_thm2_a}
	\mybigspace
\end{figure}

\begin{proof}
Consider state $s$ that is $n^2$-deterministic and assume for contradiction that $s$ is not 
deterministic.
Let $N>n^2$ be the smallest number such that $s$ is $N$-branching, and thus not $(N-1)$-branching.
Then there exist two paths $\pi = s_1,s_2,\ldots,s_N$ and $\pi'=s_1,s_2',\ldots, s_N'$ such that $s_1 = s$ and
for $i=1,2,\ldots,N-1$, we have $L(s_i)=L(s_i')$ and $L(s_N)\neq L(s_N')$.
Looking at a sequence of pairs $(s_1,s_1), (s_2,s_2'), \ldots, (s_{N-1},s_{N-1}')$, since there are at most $n^2$ possible pairs of states over $S$, by the pigeon-hole principle at least two pairs will be repeating in the observed sequence, say $(s_i,s_i')=(s_j,s_j')$, where $i<j$. 
But then the paths $\pi'' = s_1,s_2,\ldots,s_i,s_{j+1},\ldots,s_N$ and $\pi'''=s_1,s_2,\ldots,s_i,s_{j+1},\ldots,s_N$ have $M<N$ states and they witness that $s_1$ is $M$-branching, which by Lemma~\ref{lem:det-branch} is in contradiction with $s$ being $(N-1)$-deterministic.
\end{proof}

\begin{lemma}
\label{lemma:prob}
If a state $s\in S$ is $k$-branching, then any word of length $k$ starting from $s$ has probability 
 at most $(1-\pmin^{k-1})$, i.e., $\forall w\in\Sigma^k~:~\pr^s(w)\leq 1-\pmin^{k-1}$.
\label{lem:bound}
\end{lemma}

To illustrate this, observe the Markov chain in Fig.~\ref{fig:ex_with_thm2_b} with leftmost initial state.
\begin{figure}[h]
\centering\myspace
\scalebox{1}{
\begin{tikzpicture}[node distance=2cm]
\node[state,initial text=] (s) at (0,0){$a$};
\node[state] (t) [right of=s] {$a$};
\node[state] (u) [right of=t] {$a$};
\node[state] (w) [right of=u] {$b$};
\path[->]
(s) edge [loop left] node[left]{$1-\pmin$} (s)
(s) edge [right] node[below]{$\pmin$} (t)
(t) edge [right] node[below]{$\pmin$} (u)
(u) edge [right] node[below]{$\pmin$} (w)
(t) edge [bend right] node[above]{$1-\pmin$} (s)
(u) edge [bend right] node[above]{$1-\pmin$} (t)
(w) edge [bend right] node[above]{$1-\pmin$} (u)
(w) edge [loop right] node[right]{$\pmin$} (w);
\end{tikzpicture}
}\myspace
\caption{Markov chain, s.t.\ $\pr(a)=\pr(aa) = \pr(aaa) =1$, $\pr(aaab)=\pmin^3$, $\pr(aaaa)=1-\pmin^3$.}
\label{fig:ex_with_thm2_b}
\mybigspace
\end{figure}

\begin{proof}
	Let $w\in\Sigma^k$.
	Since $s$ is $k$-branching, there exists a word $w'\in\Sigma^k$ such that
	$w'\neq w$ and $\pr^s(w')>0$.
	Hence there exists at least one path with $k-1$ transitions, producing the trace $w'$, and thus $\pr^s(w')\geq \pmin^{k-1}$. Finally, $\pr^s(w)\leq 1-\pr^s(w') \leq 1-\pmin^{k-1}$. 
\end{proof}





We show that, for estimating the finite trace distance with the required precision $\epsilon$, it suffices to infer probabilities of the words up to some finite length $k$, which depends on $\epsilon$.
The idea is that paths that become deterministic before step $k$ 
do not change their probability afterwards, while
all other paths together have the probability bounded by $\epsilon$. 

\begin{lemma}
\label{lemma:sth}
  Let $s$ be a $n^2$-deterministic state in a Markov chain $\Mc$ with $n$ states. Then
  there are words $u,z$, such that $|z|+|u|\leq n$, $|u|\geq 1$, and
  \( \pr^s(z u^\omega) = 1~. \)
\end{lemma}

This motivates the following definition, where $\pref{w}$ denotes the set of all prefixes of the ($\omega$-)word $w$.

\begin{definition}
\label{def:ep}
A word $w\in \Sigma^+$ is called \emph{$k$-ultimately periodic in a Markov chain $\Mc$} if $\pr(w)> 0$ and there exists a word $u$ such that
$\word \in \pref{\Sigma^k u^\omega}$ and $1\leq |u|\leq n$, where $n$ is the number of states in $\Mc$.\qed
\end{definition}

Intuitively, for sufficiently long word $w$ and large $\epsilon$, if $\pr(w)>\epsilon$ and $w$ is $k$-ultimately periodic, then it enters within $k$ steps a BSCC, which is bisimilar to a cycle (all transition probabilities are $1$). One can also prove that this is the only way for a $\omega$-word to achieve a probability greater than $\epsilon$.

\medskip

For a word $w$ we write $B^k(w)$ for the set of paths that are labelled by $w$, have a positive probability and where all states up to step $k$ are $n^2$-branching:
  \[ B^k(w) = \{ \ppath = s_1\cdots s_{|w|}\in \Lab^{-1}(w)~|~\pr(\ppath)>0 \land \forall i\leq \min (k,|w|).~s_i \text{ is $n^2$-branching} \}~. \]
In a similar way, we write $D^k(w)$ for the set of paths that enter a ($n^2$-)deterministic state before step $k$
  \[ D^k(w) = \{ \ppath=s_1\cdots s_{|w|}\in \Lab^{-1}(w)~|~\pr(\ppath)>0 \land \exists i\leq \min (k,|v|).~s_i \text{ is $n^2$-deterministic} \}~. \]
For any $k$, we can partition paths labeled by $w$ into $B^k$-paths and $D^k$-paths:
\begin{align}
\pr(w) &= \sum_{\pi\in \Lab^{-1}(w)} \pr(\pi) 
= \sum_{\pi\in B^k(w)} \pr(\pi)  +\sum_{\pi\in D^k(w)} \pr(\pi)  \label{lep0}~.
\end{align}

Now we show that the probability of $B^k$-paths diminishes exponentially with length $k$:
\begin{lemma}
\label{lemma:dim}
  Consider a Markov chain $\Mc$ with $n$ states. For every $k\in \Nset$ and word $w$, if $|w|> k$ then
  \[ \sum_{\pi
  	\in B^{k}(w)} \pr(\pi
  ) \leq (1-\pmin^{n^2})^{\lfloor \frac{k}{n^2}\rfloor}~. \]
\end{lemma}

\begin{lemma}
\label{lemma:ep}
Let $w$ be a word in a Markov chain $\Mc$ with $n$ states. 
For every $\epsilon>0$,
if $\pr(w) > \epsilon$ and $|w|>k$ then $w$ is $k$-ultimately periodic in $\Mc$, where
$k=n^2\lceil\frac{\log \epsilon}{\log (1-\pmin^{n^2})}\rceil+n$.
\end{lemma}
\begin{proof}
  Assume that $|w|> k$. We
  split paths labelled by $w$ into $B^{k-n}(w)$ and $D^{k-n}(w)$ as in (\ref{lep0}):
  \begin{align}
    \pr(w) &= \sum_{s_1\cdots s_{|w|}\in \Lab^{-1}(w)} \pr(s_1\cdots s_{|w|}) 
       = \sum_{\substack{s_1\cdots s_{|w|}\in\\ B^{k-n}(w)}} \pr(s_1\cdots s_{|w|})  +\sum_{\substack{s_1\cdots s_{|w|}\in\\ D^{k-n}(w)}} \pr(s_1\cdots s_{|w|})
  \label{lep1}~.
  \end{align}
  By Lemma~\ref{lemma:dim} we get
  \begin{equation}
\label{lep2}
    \sum_{s_1\cdots s_{|w|}\in B^{k-n}(w)} \pr(s_1\cdots s_{|w|}) \leq \epsilon~.
  \end{equation}
  Now, from the assumption $\pr(w)> \epsilon$, \eqref{lep1} and \eqref{lep2}, it follows that
  \[ \sum_{s_1\cdots s_{|w|}\in D^{k-n}(w)} \pr(s_1\cdots s_{|w|}) > 0~. \]
This implies that there is a path $\ppath = s_1\cdots s_{|w|}\in D^{k-n}(w)$. By definition of $D^{k-n}(w)$, $\ppath$ has a $n^2$-deterministic state before step $k-n$, and
w.l.o.g.\ let $s_{k-n}$ be that state. By Lemma~\ref{lemma:sth}, every positive word from state $s_{k-n}$ is a prefix of $zu^\omega$ for some words $z,u$ such that $|z|+|u|\leq n$.
Therefore $w\in \pref{yzu^\omega}$, where $y=\Lab(s_1\cdots s_{k-n})$, i.e. $w$ is $|k|$-ultimately periodic.
\end{proof}

\begin{lemma}
\label{lemma:ep2}
Consider a Markov chain $\Mc$ with $n$ states. Let $w$ be a $k$-ultimately periodic word in $\Mc$, and $x$ be a prefix of $w$ such that $|x|> k+n$. Then \myspace

  \[ \pr(x) - \pr(w) \leq  (1-\pmin^{n^2})^{\lfloor \frac{k-n}{n^2}\rfloor}~.\]
\end{lemma}
\myspace

\begin{theorem}
\label{thm:imp}
  Consider Markov chains $\Mc_1$ and $\Mc_2$ that have at most $n$ states.
  For $\epsilon > 0$ it holds that
  \mybigspace\myspace
  
  \[ |\dists(\Mc_1,\Mc_2)-\max_{i\leq k}\dists^i (\Mc_1,\Mc_2)|\leq \epsilon,\quad\text{where}\quad k=n^2\lceil\frac{\log \epsilon}{\log (1-\pmin^{n^2})}\rceil+2n.\]
\end{theorem}
\begin{proof}
  We show that for any word $w\in \Sigma^+$:\mybigspace\myspace
  
  \begin{equation}
    \label{eq:it1} 
   \Big||\pr_1(w)-\pr_2(w)| - |\pr_1(w\downarrow k) - \pr_2(w\downarrow k)| \Big| \leq \epsilon~.   
  \end{equation}\mybigspace
  
  \noindent
  For $|w|\leq k$ \eqref{eq:it1} holds trivially.
  Suppose that $|w|\geq k$ and consider two cases.
  \begin{enumerate}
  \item If $\pr_i(w\downarrow k)> \epsilon$, then by Lemma~\ref{lemma:ep}
$w\downarrow k$ is $(k-n)$-ultimately periodic. Then by Lemma~\ref{lemma:ep2} $\pr_i(w\downarrow k)\leq \pr_i(w)+\epsilon$.
\item If $\pr_i(w\downarrow k)\leq \epsilon$, then clearly $\pr_i(w\downarrow k)\leq \pr_i(w)+ \epsilon.$
  \end{enumerate}
Both cases can be summarised by\mybigspace

\begin{equation}
\label{eq:it2}
\pr_i(w)\leq \pr_i(w\downarrow k)\leq \pr_i(w)+\epsilon~.
\end{equation}
\mybigspace

\noindent
W.l.o.g assume that $\pr_1(w)\geq \pr_2(w)$. Then by \eqref{eq:it2}\mybigspace

\[ \pr_1(w\downarrow k)-\pr_2(w\downarrow k)\geq \pr_1(w) - \pr_2(w)-\epsilon, \]
\mybigspace

\noindent
which implies \eqref{eq:it1}. 
\end{proof}

\myspace\myspace


\myspace

\section{Consequences and Discussion} 
\label{sec:disc}
\myspace

We now discuss the consequences of the (in)estimability results for several
specific subclasses of $\omega$-regular languages, captured topologically, logically, or by automata. 
We also remark on the estimability in case when the transition probabilities have finite precision.
\myspace\myspace

\subsection{Topology}
\myspace

\noindent\textbf{Negative result for clopen sets}
Note that the proof of inestimability was based on the ability to express the events $E_n$ for any $n\in\Nset$:
\begin{quote}
$E_n=$ ``there is at most $c_n=(0.5+\tau/2)n$ symbols $b$ in the prefix path of length $n$.''
\end{quote}
Observe that each $E_n$ can be expressed as finite union of cones, each expressing exact positions of $a$'s and $b$'s in the first $n$ steps.
For instance, for $\tau =0.2$, the event $E_{2}$, ``there is at most $1$ symbol $b$ in the first $2$ steps,''  can be described by the union $\cone(aa)\cup\cone(ab)\cup\cone(ba)$.

Since finite unions of cones form exactly the clopen sets, the lowest class $\Delta_1$ in the Borel hierarchy, it follows that distances based on any class in the hierarchy are inestimable.
\smallskip

\noindent\textbf{Positive result for the infinite-trace distance} Using the result on finite-trace distance, we can prove that the infinite-trace distance $\disti$ of Example~\ref{ex:trac-dist} is also estimable. 
Indeed, the distance is non-zero only due to $k$-ultimately periodic $\omega$-words with positive probability.
By Lemma~\ref{lemma:ep2} we can provide confidence intervals for these probabilities through the $k$-prefixes using the fixed-length distance $\dists^k$.
\myspace\myspace

\subsection{Logic}
\myspace

\noindent\textbf{Negative result for LTL.} 
The LTL distance as in Example~\ref{ex:log-dist} is again inestimable since we can express the event $E_n$ in LTL by a finite composition of operators $\X,\wedge,\vee$ (notably this fragment induces the same distance as LTL~\cite{DBLP:conf/ictac/BacciBLM15}).
Indeed, 
for instance, for $\tau =0.2$, the event $E_{10}$, ``there is at most $6$ symbols $b$ in the path prefix of length $10$,'' is equivalent to ``at least $4$ symbols $a$ in the path prefix of length $n$,'' and it can be described by a disjunction of ${10\choose 4}$ formulae, each determining the possible position of symbols $a$, resulting in a formula 
$(a\wedge \X a\wedge \X^2 a\wedge\X^3 a)\vee (a\wedge \X a\wedge \X^2 a\wedge \X^4 a)\vee\ldots\vee (\X^7 a \wedge\X^8 a \wedge \X^9 a \wedge\X^{10} a)$.
\smallskip

\noindent\textbf{Positive result for \LTL(\F\G,\G\F).} The distance generated by the fragment of LTL described by combining operators $\F\G$ and $\G\F$ and Boolean operators is estimable. 
Notice that the probability of the property  $\varphi\equiv\F\G\varphi'$ equals the probability of reaching a BSCC such that $\varphi'$ holds in all of its states, 
while the probability of property $\varphi\equiv\G\F\varphi'$ equals the probability that every BSCC contains a state which satisfies $\varphi'$. 
Hence, properties expressed in this fragment of LTL can be checked by inferring all BSCCs of a chain and a simple analysis of them. 
The statistical estimation of all BSCCs for labelled Markov chains where only the minimal transition probability is known is possible and is shown in \cite{DBLP:journals/corr/DacaHKP15}.
\myspace\myspace

\subsection{Automata}
\myspace

\noindent\textbf{Negative result for automata distances.} 
For the class of all deterministic Rabin automata (DRA), the distance (as in Example~\ref{ex:auto-dist}) is inestimable. 
This is implied by the inestimability for clopen sets or for LTL. Further, we can also directly encode the event $E_n$ that ``at least $k$ symbols $a$ are observed in the path of length $n$'' by an automaton: 
the DRA counts how many symbols $a$ are seen in the prefix up to length $n$; this can be done with $k\cdot n$ states where the automaton is in a state $s_{k',n'}$ if and only if in the $n'\leq n$ prefix of the input word, there are $k'\leq k$ symbols $a$. 
\smallskip

\noindent\textbf{Positive result for fixed-size automata.}
When restricting to the class of DRA of size at most $k\in \Nset$, the distance $\dist_{DRA\leq k}$ can be estimated.
A naive algorithm amounts to enumerating all automata up to given size $k$, then applying statistical model checking to infer the probability of satisfying the automata in each of the Markov chains, and checking for which automaton the probability difference in the two chains is maximized. 
Statistically inferring the probability of whether a (black-box) Markov chain satisfies a property given by a DRA is a subroutine of the procedure for statistical model checking Markov chains for LTL, described in \cite{DBLP:journals/corr/DacaHKP15}.
\myspace\myspace

\subsection{Finite Precision}
\myspace

When the transition probabilities have finite precision, e.g.\ are given by at most two decimal digits, several negative results turn positive. 
Finite precision allows us to learn the MCs exactly with high probability, by rounding the learnt transition probabilities to the closest multiple of the precision.
Subsequently, we can approximate the distance by the algorithms applicable in the white-box setting.
In case of the total variation distance, one can apply the approximation algorithm of \cite{DBLP:conf/csl/ChenK14}; for trace distances, the approximation algorithm of \cite{DBLP:conf/ictac/BacciBLM15} is also available.
In particular, for the special case of the trace equivalence $\diste$ we can leverage the fact that Markov chains are equivalent when all their traces up to length $|\Mc_1| + |\Mc_2|-1$ have equal probability.
With the assumption of finite precision one can get by sampling the exact distribution of such traces with high confidence.
Note that the same algorithm can not be applied without assuming finite precision, since arbitrarily small difference in chains cannot be detected. 
 
%

%

\myspace\myspace

\section{Conclusions and Future Work}
\label{sec:concl}
\myspace

We have introduced a linear-distance framework for Markov chains and considered estimating the distances in the black-box setting from simulation runs.
We investigated several distances, delimiting the (in)estimability  boarder for distances given topologically, logically, and by automata.
As the next step, it is desirable to look for practical algorithms that would converge fast on practical benchmarks. 
Another direction is to characterize the largest language for which the distance can be estimated, and, dually, the smallest language that cannot be estimated.

\myspace

\bibliographystyle{plainurl}
\bibliography{ref}

\newpage
\appendix
\section{Proofs from Section~\ref{sec:neg}}

We show that $q$ as defined in the proof of Theorem~\ref{thm:tv} is finite.

\begin{claim}
	$q\in \Nset$.
\end{claim}
\begin{proof}
 Suppose for contradiction that $q=\infty$, then
\begin{equation}
\label{eq:n1}
\forall c\in \Nset ~.~\pr_1(Y \leq c)< 0.5+\alpha
\end{equation}
From the standard results in probability theory we obtain
\begin{equation}
\label{eq:n2}
\lim_{c\rightarrow \infty^-} \pr_1(Y\leq c) =  \pr_1(Y\in \Nset).
\end{equation}
From the assumption that the algorithm terminates almost surely we get that the RHS of \eqref{eq:n2} equals $1$, while the LHS
must be $\leq 0.5+\alpha<1$ by \eqref{eq:n1}, which is a contradiction.  
\end{proof}

\bigskip

\noindent
We show that $\epsilon$ as defined in the proof of Theorem~\ref{thm:tv} is positive.

\begin{claim}
	$\epsilon>0$.
\end{claim}
\begin{proof}
 
For $q=0$ this is trivial. Otherwise,
observe that the term
$\alpha^\frac{1}{q}2^\frac{1-q}{q}$ is monotonically increasing in
$\alpha$. Thus,
\[ \alpha^\frac{1}{q}2^\frac{1-q}{q} <
0.5^\frac{1}{q}2^\frac{1-q}{q} = 0.5, \]
which implies that $\epsilon >0$.
\end{proof}

\section{Proofs from Section~\ref{sec:pos}}

\begin{reflemma}
	{lem:det-branch}
	If $s\in S$ is $k$-branching, it is then $(k+1)$-branching. 
	Dually, if it is $k$-deterministic, it is also $(k-1)$-deterministic.
\end{reflemma}

\begin{proof}
	The lemma follows trivially from the definition: if there exist two different words $w,w'\in \Sigma^k$ such that $\pr^s(w)>0$ and $\pr^s(w')>0$, they can be always extended to different words $wa, w'a'\in \Sigma^{k+1}$ with positive probability. 
\end{proof}

\begin{reflemma}
	{lemma:sth}
	Let $s$ be a $n^2$-deterministic state in Markov chain $\Mc$ with $n$ states. Then
	there are words $u,z$, such that $|z|+|u|\leq n$, $|u|\geq 1$, and $\pr^s(z u^\omega) = 1$.
\end{reflemma}

\begin{proof}
	Consider any run $\run=s_1 s_2 \cdots$, where $s_1 = s$.  Let $t$ be
	the first state on $\run$ that occurs twice, i.e.\ $j$ is the
	smallest index such that
	\[ \exists i.~i<j~\land~s_i=s_j=t~.\]
	Here $s_i, s_j$ are the first and second occurrence of $t$ on
	$\run$, respectively.  It holds that $j\leq n+1$, because
	otherwise  some other state would occur twice earlier than $t$.
	
	Let $u,z$ be the following words $z = \Lab(s_1\cdots s_{i-1})$ and $u = \Lab(s_i\ldots s_{j-1}).$
	Clearly $|z|+|u|\leq n$ and $|u|\geq 1$.  The word $u$ can be repeated any number
	of times from state $s_i=s_j=t$.  By Lemma \ref{lemma:det}, state
	$s$ is $\infty$-deterministic and thus $t$ as well. Hence $u^\omega$ has the
	probability one from the state $t$.  As a consequence
	$ \pr_s(z u^\omega) = 1$.
\end{proof}

\begin{reflemma}
	{lemma:dim}
	Consider a Markov chain $\Mc$ with $n$ states. For every $k\in \Nset$ and word $w$, if $|w|> k$ then
	\[ \sum_{\pi
		\in B^{k}(w)} \pr(\pi
	) \leq (1-\pmin^{n^2})^{\lfloor \frac{k}{n^2}\rfloor}~. \]
\end{reflemma}

\begin{proof}
	Through the proof let $c=n^2$ and $w\uparrow n$ denote the suffix of $w$ of length $n$. We show
	\begin{align*}
	&\sum_{s_1\cdots s_{|w|}\in B^k(w)} \pr(s_1\cdots s_{|w|}) \\
	&\leq\sum_{s_1\cdots s_{k}\in B^k(w)} \pr(s_1\cdots s_{k})
	& \substack{s_1\cdots s_k=s_1\cdots s_{|w|}\downarrow k } \\
	&= \sum_{\substack{s_1\cdots s_{k-c} \in\\ B^k(w\downarrow k-c)}} \pr(s_1\cdots s_{k-c})
	\left( \sum_{\substack{s_{k-c}\cdots s_{k}\in\\B^c(w\uparrow c+1)}} \pr^{s_{k-c}}(s_{k-c}\cdots s_{k})\right) 
	& \substack{\text{split $w$ into $w\downarrow k-c$}\\\text{ and $w\uparrow (c+1)$}}\\
	&\leq \sum_{\substack{s_1\cdots s_{k-c} \in\\ B^k(w\downarrow (k-c))}} \pr(s_1\cdots s_{k-c})
	\cdot\pr^{s_{k-c}}(w\uparrow c+1)
	& \substack{ B^c(x)\subseteq L^{-1}(x) }\\
	&\leq \sum_{\substack{s_1\cdots s_{k-c}\in \\ B^k(w\downarrow k-c)}} \pr(s_1\cdots s_{k-c}) (1-\pmin^{c}) & \substack{\text{by Lemma~\ref{lemma:prob}, since }\\s_{k-c}\text{ is $c$-branching,}\\ \text{and thus $(c+1)$-branching}}\\
	&\leq \sum_{\substack{s_1\cdots s_{k-2c}\in \\ B^k(w\downarrow k-2c)}} \pr(s_1\cdots s_{k-2c}) (1-\pmin^{c})^2 & \substack{\text{by Lemma~\ref{lemma:prob}, since }\\s_{|w|-2c}\text{ is $c$-branching,}\\ \text{and thus $(c+1)$-branching}}\\
	&\leq \sum_{\substack{s_1\cdots s_{k-\lfloor\frac{k}{c}\rfloor c}\in \\ B^k(w\downarrow k-\lfloor\frac{k}{c}\rfloor c)}} \pr(s_1\cdots s_k) (1-\pmin^{c})^{\lfloor\frac{k}{c}\rfloor} & \substack{\text{by repeatedly}\\\text{applying Lemma~\ref{lemma:prob}}}\\
	&\leq (1-\pmin^{c})^{\lfloor \frac{k}{c}\rfloor}~.
	\end{align*}
\end{proof}

\begin{reflemma}
	{lemma:ep2}
	Consider a Markov chain $\Mc$ with $n$ states. Let $w$ be a $k$-ultimately periodic word in $\Mc$, and $x$ be a prefix of $w$ such that $|x|> k+n$. Then 
	\[ \pr(x) - \pr(w) \leq  (1-\pmin^{n^2})^{\lfloor \frac{k-n}{n^2}\rfloor}~.\]
\end{reflemma}

\begin{proof}
	Let $c=n^2$.
	We split $\pr(x)$ and $\pr(w)$ in the following way:
	\begin{align}
	\pr(x) &= \overbrace{\sum_{s_1\cdots s_{|x|}\in B^{k-n}(x)} \pr(s_1\cdots s_{|x|})}^{SB_1}  +\overbrace{\sum_{s_1\cdots s_{|x|}\in D^{k-n}(x)} \pr(s_1\cdots s_{|x|})}^{SD_1}  \label{lep2-1}\\
	\pr(w) &= \underbrace{\sum_{s_1\cdots s_{|w|}\in B^{k-n}(w)} \pr(s_1\cdots s_{|w|})}_{SB_2}  +\underbrace{\sum_{s_1\cdots s_{|w|}\in D^{k-n}(w)} \pr(s_1\cdots s_{|w|})}_{SD_2}  \label{lep2-2}~.
	\end{align}
	By Lemma~\ref{lemma:dim} we get
	\begin{align}
	&SB_1\leq (1-\pmin^{c})^{\lfloor \frac{k-n}{c}\rfloor} \label{lep2-3}\\
	&SB_2 \leq (1-\pmin^{c})^{\lfloor \frac{k-n}{c}\rfloor}~. \label{lep2-4} 
	\end{align}
	We now prove that the deterministic paths for $w$ and $x$ have the same probability
	\begin{equation}
	\label{lep2-5}
	SD_1 = SD_2.
	\end{equation}
	Consider any path $\ppath=s_1\cdots s_{|x|}\in D^{k-n}(x)$.  By
	definition path $\ppath$ enters a $n^2$-deterministic state before
	step $k-n$. W.l.o.g let $s_{k-n}$ be that $n^2$-deterministic
	state. By Lemma~\ref{lemma:sth} there are words $z,u$ such that
	\[ \pr_{s_{k-n}}(zu^\omega) = 1, \]
	and $|u|\geq 1$, $|z|+|u|\leq n$. 
	Thus the word labelling $\ppath$ has the form
	\[x = \Lab(\ppath) \in \pref{y z u^\omega},\]
	where $y = \Lab(s_1 \cdots s_{k-n})$.
	Both $x$ and $w$ are $k$-ultimately periodic and have length greater than $k+n$, so they both must be of the form
	\[w,x = \Lab(\ppath) \in \pref{y z u^\omega}~.\]

        Consider any path $\ppath \in D^{k-n}(x)$, and let $E(\ppath)$ denote all extensions of $\ppath$ to the paths of length $|w|$:
	\[ E(\ppath) = \{ p_1\cdots p_{|w|}~|~\forall i\leq |x|.~ p_i=s_i \land \forall i<|w|.~\Pm(p_i,p_{i+1})> 0\}~.\]
	All paths in $E(\ppath)$ are labelled by the same word, namely $w$, and enter a $n^2$-deterministic state before step $k-n$, therefore
$E(\ppath)\subseteq D^{k-n}(w)$, which  implies that $\pr(SB_1) \leq \pr(SB_2)$.
Now, consider any path $\ppath \in D^{k-n}(w)$. The prefix of $\ppath \downarrow |x|$ is labelled by the word $x$, and enters a $n^2$-deterministic state before step $k-n\geq |x|$, so 
$\ppath\downarrow |x| \in D^{k-n}(w)$; this implies the other inequality  that $\pr(SB_2) \leq \pr(SB_1)$.

	
	Finally, we write 
	\begin{align*}
	&\pr(x) - \pr(w)=  SB_1 + SD_1 - SB_2 - SD_2 &\text{by \eqref{lep2-1} and \eqref{lep2-2}}\\
	&= SB_1 - SB_2 & \text{by \eqref{lep2-5}}\\
	&\leq (1-\pmin^{c})^{\lfloor \frac{k-n}{c}\rfloor} &\text{by \eqref{lep2-3} and \eqref{lep2-4}}~.
	\end{align*}
\end{proof}


\end{document}